\documentclass{llncs}

\usepackage{amsmath}
\usepackage{xspace}
\usepackage{amssymb}
\usepackage{epsfig}

\newcommand\mbZ{\mbox{$\mathbb{Z}$}}
\newcommand\mbC{\mbox{$\mathbb{C}$}}

\newcommand {\ie} {\textit{i.e.}\xspace}
\newcommand {\st} {\textit{s.t.}\xspace}

\newcommand{\Z}{\mathbb{Z}}

\newcommand{\R}{\mathbb{R}}

\def\01{\{0,1\}}
\newcommand{\sn}[1]{\|#1\|}
\newcommand{\fro}[1]{\|#1\|_F}
\newcommand{\trn}[1]{\|#1\|_{tr}}
\newcommand{\Tr}{\mathrm{Tr}}
\newcommand{\rk}{\mathrm{rk}}
\newcommand{\size}{\mathrm{size}}
\newcommand{\disc}{\mathrm{disc}}
\newcommand{\hp}{\bullet}
\newcommand{\comp}{\circ}

\newcommand{\PARITY}{\mathrm{PARITY}}

\newcommand{\DT}{\mathrm{DT}}
\newcommand{\RT}{\mathrm{RT}}
\newcommand{\QT}{\mathrm{QT}}

\newcommand{\braket}[2]{\langle#1, #2\rangle}

\newtheorem{fact}[theorem]{Fact}

\begin{document}
\title{Composition theorems in communication complexity}
\date{}
\author{Troy Lee\inst{1} \and Shengyu Zhang\inst{2}}
\institute{Rutgers University, \email{troyjlee@gmail.com}  \and The Chinese University of Hong Kong,  \email{syzhang@cse.cuhk.edu.hk}}
\maketitle

\begin{abstract}
A well-studied class of functions in communication complexity are
composed functions of the form $(f \comp g^n)
(x,y)=f(g(x^1, y^1), \ldots, g(x^n,y^n))$.  This is a rich family of 
functions which encompasses many of the important examples in the literature.  
It is thus of great interest to understand what properties of $f$ and $g$ affect the communication
complexity of $(f \comp g^n)$, and in what way. 

Recently, Sherstov \cite{She09b} and independently Shi-Zhu \cite{SZ09b} 
developed conditions on the inner function $g$ which imply that the quantum 
communication complexity of $f \comp g^n$ is at least the approximate 
polynomial degree of $f$.  We generalize both of these frameworks.  We show 
that the pattern matrix framework of Sherstov works whenever the inner function
$g$ is {\em strongly balanced}---we say that $g: X \times Y \rightarrow \{-1,+1\}$ is 
strongly balanced if all 
rows and columns in the matrix $M_g=[g(x,y)]_{x,y}$ sum to zero.  
This result strictly generalizes the pattern matrix framework of Sherstov \cite{She09b},
which has been a very useful idea in a variety of settings
\cite{She08b,RS08,Cha07,LS09,CA08,BHN09}.  

Shi-Zhu require that the inner function $g$ has small {\em spectral discrepancy}, 
a somewhat awkward condition to verify.  We relax this to the usual notion of discrepancy.

We also enhance the framework of composed functions studied so far
by considering functions $F(x,y) = f(g(x,y))$, where the range of $g$ is a 
group $G$.  When $G$ is Abelian, the analogue of the strongly balanced 
condition becomes a simple group invariance property of $g$.  We are able to 
formulate a general lower bound on $F$ whenever $g$ satisfies this 
property.
\end{abstract}

\section{Introduction}
Communication complexity studies the minimum amount of
communication needed to compute a function whose input variables
are distributed between two or more parties.  Since the
introduction by Yao \cite{Yao79} of an elegant mathematical model
to study this question, communication complexity has grown into a
rich field both because of its inherent mathematical interest and
also its application to many other models of computation.  See the
textbook of Kushilevitz and Nisan \cite{KN97} for a comprehensive
introduction to the field.

In analogy with traditional computational complexity classes, one
can consider different models of communication complexity based on
the resources available to the parties.  Besides the standard
deterministic model, of greatest interest to us will be a
randomized version of communication complexity, where the parties
have access to a source of randomness and are allowed to err with
some small constant probability, and a quantum model where the
parties share a quantum channel and the cost is measured in
qubits.

Several major open questions in communication complexity ask about
how different complexity measures relate to each other.  The log rank
conjecture, formulated by Lov\'{a}sz and Saks \cite{LS88}, asks if
the deterministic communication complexity of a Boolean function
$F : X \times Y \rightarrow \01$ is upper bounded by a polynomial
in the logarithm of the rank of the matrix $[F(x,y)]_{x,y}$.
Another major open question is if randomized and quantum
communication complexity are polynomially related for all total
functions.  We should mention here that the assumption of the
function being total is crucial as an exponential separation is
known for a partial function \cite{Raz99}.

One approach to these questions has been to study them for
restricted classes of functions. Many functions of interest are
{\em block composed} functions.  For finite sets $X,Y$, and $E$, a  
function $f: E^n \rightarrow \{-1,+1\}$, and a function $g : X \times Y \rightarrow
E$, the block composition of $f$ and $g$ is the function $f \comp
g^n: X^n \times Y^n \to \{-1,+1\}$ defined by $(f \comp
g^n)(x,y)=f(g(x^1, y^1), \ldots, g(x^n, y^n))$ where $(x^i,y^i)
\in X \times Y$ for all $i=1, \ldots, n$. For example, if $E=\{-1,+1\}$, 
the inner product function results when $f$ is PARITY and $g$ is AND,
set-intersection when $f$ is OR and $g$ is AND, and the equality
function when $f$ is AND and $g$ is the function IS-EQUAL, which
is one if and only if $x=y$.

In a seminal paper, Razborov \cite{Raz03} gave tight bounds for
the bounded-error quantum communication complexity of block composed
functions where the outer function $f$ is symmetric and the inner
function $g$ is bitwise AND.  In particular, this result showed that
randomized and quantum communication complexity are polynomially
related for such functions.

More recently, very nice frameworks have been developed by
Sherstov \cite{She07,She09b} and independently by Shi and Zhu
\cite{SZ09b} to bound the quantum complexity of block composed
functions that goes beyond the case of symmetric $f$ to work for
any $f$ provided the inner function $g$ satisfies
certain technical conditions. When $g$ satisfies these conditions, 
this framework allows one to lower bound the quantum communication 
complexity of $f \comp g^n$ in terms of the approximate polynomial 
degree of $f$, a classically well-studied measure.  Shi and Zhu are 
able to get a bound on $f \comp g^n$ in terms of the approximate 
degree of $f$ whenever $g$ is sufficiently ``hard''---unfortunately, the 
hardness condition they need is in terms of ``spectral discrepancy,'' 
a quantity which is somewhat difficult to bound, and their bound 
requires that $g$ is a function on at least $\Omega(\log(n/d))$ bits,  where $d$ is the
approximate polynomial degree of $f$. Because of this, Shi-Zhu are
only able to reproduce Razborov's results with a polynomially
weaker bound.

Sherstov developed so-called {\it pattern matrices} which are the
matrix representation of a block composed function when $g$ is a
fixed function of a particularly nice form.  Namely, in a pattern
matrix the inner function $g: \{-1,+1\}^k \times ([k] \times
\{-1,+1\}) \to \{-1,+1\}$ is parameterized by a positive integer
$k$ and defined by $g(x,(i,b))=x_i \cdot b$, where $x_i$ denotes
the $i^{th}$ bit of $x$. In other words, the first argument of $g$ is a $k$ 
bit string $x$, and the second argument selects a bit of $x$ or its negation. 
So here $X=\{-1,+1\}^k$, $Y = [k] \times
\{-1,+1\}$ and the intermediate set $E$ is $\{-1,+1\}$.  With this $g$, 
Sherstov shows that the approximate polynomial degree of $f$ is a lower 
bound on the quantum communication complexity of $f \comp g^n$,
for any function $f$. Though seemingly quite special, pattern matrices 
have proven to be an extremely useful concept.  First, they give a simple proof of
Razborov's tight lower bounds for $f(x\wedge y)$ for symmetric
$f$. Second, they have also found many other applications in
unbounded-error communication complexity \cite{She08b,RS08} and
have been successfully extended to multiparty communication
complexity \cite{Cha07,LS09,CA08,BHN09}.

A key step in both the works of Sherstov and Shi-Zhu is to bound
the spectral norm of a sum of matrices $\sn{\sum_i {B_i}}$.  This
is the major step where these works differ.  Shi-Zhu apply the
triangle inequality to bound this as $\sn{\sum_i {B_i}} \leq
\sum_i \sn{B_i}$. On the other hand, Sherstov observes that in the case of
pattern matrices the terms of this sum are mutually orthogonal,
\ie $B_i^\dag B_j = B_i B_j^\dag = 0$ for all $i\neq j$.  In this case,
one has a stronger bound on the spectral norm $\sn{\sum_i {B_i}} =
\max_i \sn{B_i}$.

In this paper, we extend both of the frameworks of Sherstov and Shi-Zhu.  
In the case of Shi-Zhu, we are able to reprove their theorem with 
the usual notion of discrepancy instead of the somewhat awkward spectral discrepancy they
use.  The main observation we make is that as all Shi-Zhu use in this step is the 
triangle inequality, we can repeat the argument with any norm here, including discrepancy itself.

In the case of pattern matrices, special properties of the spectral norm are used, 
namely the fact about the spectral norm of a sum of orthogonal matrices.  
We step back to see what key features of a pattern matrix lead to this 
orthogonality property.
We begin with the Boolean case, that is, where the intermediate set $E$ is taken 
to be $\{-1,+1\}$.  In this case, a crucial concept is the notion of a {\em strongly 
balanced} function.  We say that $g : X \times Y \rightarrow \{-1,+1\}$ is strongly
balanced if in the sign matrix $M_g[x,y]=g(x,y)$ all rows and all columns sum to zero.
We show that whenever the inner function $g$ is strongly balanced, the key orthogonality 
condition holds; this implies that whenever $g$ is strongly balanced and the communication 
matrix of $g$ has rank larger than one, the approximate degree of the outer function $f$ is a lower 
bound on the quantum communication complexity of $f \comp g^n$.  

The requirement that the communication matrix of $g$ has rank larger than one is necessary for
such a statement.  For example, when $g(x,y)=\oplus(x,y)$ is the 
XOR function on one bit, then the communication complexity of $\PARITY \comp g^n$ 
is constant, while PARITY has linear approximate polynomial degree.  It turns out that when 
$g$ is rank-one, the appropriate measure of the complexity of $f \comp g^n$ is no longer the 
approximate degree of $f$, but the minimum $\ell_1$ norm of Fourier coefficients of a function 
entrywise close to $f$; see the survey \cite{LS09d} for a description of this case.  

We also consider the general case where the intermediate set is any group $G$. That is, we 
consider functions $F(x,y) = f(g(x,y))$, where $g: X \times Y \rightarrow G$ for a group $G$ and 
$f: G \rightarrow \{-1, +1\}$ is a class function on $G$. The case $E = \{-1,+1\}$ discussed above 
corresponds to taking the group $G=\Z_2^n$.  When $G$ is a general Abelian group, the key 
orthogonality condition requires more than that the matrix $M_g[x,y]=g(x,y)$ is strongly balanced; 
still, it admits a nice characterization in terms of group invariance.  A multiset $T\in G\times G$ is 
said to be $G$-invariant if $(s,s)T = T$ for all $s\in G$. The orthogonality condition will hold if and 
only if all pairs of rows and all pairs of columns of $M_g$ (when viewed as multisets) are $G$-
invariant. One can generalize the results discussed above to this general setting with appropriate 
modifications. In the case that $G = \Z_2^n$, the $G$-invariant condition degenerates to the 
strongly balanced requirement of $M_g$.

\section{Preliminaries}
\label{sec:main_preliminaries}
All logarithms are base two.  For a complex number $z=a+ib$ we let 
$\bar z=a- ib$ denote the complex conjugate of $z$ and $|z|=\sqrt{a^2+b^2}$ 
and $\mathrm{Re}(z)=a$.  

\subsection{Complexity measures}
We will make use of several complexity measures of functions and
matrices. Let $f: \{-1,+1\}^n \rightarrow \{-1,+1\}$ be a
function. For $T \subseteq \01^n$, the Fourier
coefficient of $f$ corresponding to the character $\chi_T$ is
$\hat f_T=\frac{1}{2^n} \sum_x f(x) \chi_T(x)=\frac{1}{2^n} \sum_x f(x) \prod_{i \in T} x_i.$
The {\em degree} of $f$ as a polynomial, denoted $\deg(f)$, is the
size of a largest set $T$ for which $\hat f_T \ne 0$. 

We will need some notations for matrices. We reserve $J$ for the all ones matrix, 
whose size will be determined by the context.  For a matrix $A$ let
$A^\dag$ denote the conjugate transpose of $A$.  We use $A \hp B$ for the 
entrywise product of $A,B$, and $A \otimes B$ for the tensor product.  
If $A$ is an $m$-by-$n$ matrix then we say that $\size(A)=mn$.  We 
use $\braket{A}{B}=\Tr(AB^\dag)$ for the inner product of $A$ and $B$.

Let $\|A\|_1$ be the $\ell_1$ norm of $A$, i.e.\ sum of the absolute values 
of entries of $A$, and $\|A\|_\infty$ the $\ell_\infty$ norm.  
For a positive semidefinite matrix $M$ let $\lambda_1(M) \ge \cdots \ge \lambda_n(M) \ge 0$ be
the eigenvalues of $M$.  We define the $i^{th}$ singular value of
$A$, denoted $\sigma_i(A)$, as
$\sigma_i(A)=\sqrt{\lambda_i(AA^\dag)}$.  The rank of $A$, denoted
$\rk(A)$ is the number of nonzero singular values of $A$.  We will
use several matrix norms. The spectral or operator norm is the largest singular value $\sn{A}=\sigma_1(A)$, the trace norm is the summation of all singular values $\trn{A}=\sum_i \sigma_i(A)$, and the Frobenius norm is the $\ell_2$ norm of the singular values $\fro{A}=\sqrt{\sum_i \sigma_i(A)^2}$.

When $AB^\dag =A^\dag B=0$ we will say that $A,B$ are {\em orthogonal}. Please note the difference with the common use of this term, which usually means $\braket{A}{B}=0$. The following facts are easily seen.
\begin{fact}\label{fact:ortho}
Let $A,B$ be two matrices of the same dimensions and suppose that
$AB^\dag=A^\dag B=0$.  Then
\begin{align*}
&\rk(A+B) =\rk(A)+\rk(B), \ \trn{A+B} =\trn{A}+\trn{B}, \ \sn{A+B} =\max \{\sn{A},\sn{B}\}.
\end{align*}
\end{fact}

Another norm we will use is the $\gamma_2$ norm, introduced to complexity theory 
in \cite{LMSS07}, and familiar in matrix analysis as the Schur product operator norm.  
The $\gamma_2$ norm can be viewed as a weighted version of the trace norm.
\begin{definition}
$$\gamma_2(A)=\max_{u,v: \|u\|=\|v\|=1} \trn{A \hp uv^\dag}.$$
\end{definition}
Here $A \hp B$ denotes the entrywise product of $A$ and $B$.  
It is clear from this definition that $\gamma_2(A) \ge \trn{A}/\sqrt{mn}$ for a 
$m$-by-$n$ matrix $A$.

For a norm $\Phi$, the dual norm $\Phi^*$ is defined as 
$\Phi^*(v)=\max_{u: \Phi(u) \le 1} |\braket{u}{v}|$.  For example, the 
$\ell_\infty$ norm is dual to the $\ell_1$ norm, and the 
spectral norm is dual to the trace norm.

The norm $\gamma_2^*$, dual to the $\gamma_2$ norm, looks as follows.
\begin{definition}
$$
\gamma_2^*(A)=\max_{\substack{u_i,v_j \\ \|u_i\|=\|v_j\|=1}} \sum_{i,j} A[i,j] \braket{u_i}{v_j}.
$$
\end{definition}

Another complexity measure we will make use of is discrepancy
\begin{definition}
Let $A$ be an $m$-by-$n$ sign matrix and let $P$ be a probability distribution on the entries 
of $A$.  The discrepancy of $A$ with respect to $P$, denoted $\disc_P(A)$, is defined as
$$
\disc_P(A)=\max_{\substack{x \in \{0,1\}^{m} \\ y \in \{0,1\}^{n}}} |x^\dag A\hp P y|.
$$
\end{definition}

We will write $\disc_U(A)$ for the special case where $P$ is the uniform distribution.
It is easy to see from this definition that
$
\disc_U(A) \le \frac{\sn{A}}{\sqrt{\size(A)}}.
$
Shaltiel \cite{Sha03} has shown the deeper result that this bound is in fact polynomially tight:
\begin{theorem}[Shaltiel]\label{thm:shaltiel}
Let $A$ be a sign matrix. Then 
$$\frac{1}{108} \left(\frac{\sn{A}}{\sqrt{\size(A)}}\right)^3 \le \disc_U(A).$$
\end{theorem}

Discrepancy and the $\gamma_2^*$ norm are very closely related.  Linial and Shraibman 
\cite{LS09b} observed that Grothendieck's inequality gives the following.
\begin{theorem}[Linial-Shraibman]
For any sign matrix $A$ and probability distribution $P$
$$
\disc_P(A) \le \gamma_2^*(A \hp P) \le K_G \; \disc_P(A)
$$
where $1.67\ldots \le K_G \le 1.78\ldots$ is Grothendieck's constant.
\end{theorem}

\subsubsection*{Approximate measures}
We will also use approximate versions of these complexity measures
which come in handy when working with bounded-error 
models.  Say that a function $g$ gives an
$\epsilon$-approximation to $f$ if $|f(x) - g(x)| \le \epsilon$
for all $x \in \{-1,+1\}^n$.  The $\epsilon$-approximate
polynomial degree of $f$, denoted $\deg_\epsilon(f)$, is the
minimum degree of a function $g$ which gives an
$\epsilon$-approximation to $f$.  

We will similarly look at the $\epsilon$-approximate version 
of the trace and $\gamma_2$ norms.  We give the general definition with respect to 
any norm.
\begin{definition}[approximation norm]
Let $\Phi: \R^n \rightarrow \R$ be an arbitrary norm.  Let 
$v \in \R^n$ be a sign vector.  For $0 \le \epsilon < 1$ we define
the approximation norm $\Phi^\epsilon$ as 
$$
\Phi^\epsilon(v)=\min_{\substack{u \\ \|v-u\|_\infty \le \epsilon}} \Phi(u).
$$
\end{definition}
Notice that an approximation norm $\Phi^\epsilon$ is not itself a norm---
we have only defined it for sign vectors, and it will in general not satisfy 
the triangle inequality.

As a norm is a convex function, using the separating hyperplane theorem
one can quite generally give the following equivalent dual formulation
of an approximation norm.
\begin{proposition}
\label{prop:dual_approximate_norm} Let $v \in \R^n$ be a sign vector, and
$0 \le \epsilon < 1$
$$
\Phi^\epsilon(v) = \max_u
\frac{|\braket{v}{u}|-\epsilon\|u\|_1}{\Phi^*(u)}
$$
\end{proposition}
A proof of this can be found in the survey \cite{LS09d}.

\subsection{Communication complexity}
Let $X,Y,S$ be finite sets and $f: X \times Y \rightarrow S$ be a function.  
We will let $D(f)$ be the deterministic communication complexity of $f$, 
and $R_\epsilon(f)$ denote the randomized public coin complexity of 
$f$ with error probability at most $\epsilon$.  We refer to the reader to 
\cite{KN97} for a formal definition of these models.  We will also study
$Q_\epsilon(f)$ and $ Q_\epsilon^*(f)$, the $\epsilon$-error quantum communication 
complexity of $f$ without and with shared entanglement, respectively.  
We refer the reader to \cite{Raz03} for a nice description of these models.

For notational convenience, we will identify a function 
$f: X \times Y \rightarrow \{-1,+1\}$ with its sign matrix $M_f = [f(x,y)]_{x,y}$.  Thus, 
for example, $\sn{f}$ refers to the spectral norm of the sign matrix representation
of $f$.  

For all of our lower bound results we will actually lower bound the approximate 
trace norm or $\gamma_2$ norm of the function.  Razborov showed that the approximate 
trace norm can be used to lower bound on quantum communication complexity, and 
Linial and Shraibman generalized this to the $\gamma_2$ norm. 
\begin{theorem}[Linial-Shraibman \cite{LS09c}] \label{thm: approx trace}
Let $A$ be a sign matrix and $0 \le \epsilon < 1/2$.  Then
$$
Q_\epsilon^*(A) \ge \log\left(\gamma_2^{2\epsilon}(A) \right) - 2.
$$
\end{theorem}

\subsubsection*{Composed functions} \vspace{-1em}
Before discussing lower bounds on a block composed function $f \comp g^n$, 
let us see what we expect the complexity of such a function to be. A
fundamental idea going back to Nisan \cite{Nis94} and Buhrman, Cleve, and Wigderson
\cite{BCW98}, is that the complexity of $f \comp g^n$ can be related
to the {\em decision tree} complexity, also known as {\em query} complexity, of $f$ and the 
communication
complexity of $g$. Let $\DT(f)$ be the query complexity of
$f$, that is the number of queries of the form $x_i=?$ needed to
evaluate $f(x)$ in the worst case.  Similarly, let $\RT_{\epsilon}(f),
\QT_{\epsilon}(f)$ denote the randomized and quantum query 
complexity of $f$ respectively, with error probability at most
$\epsilon$.  For formal definitions of these measures and a survey of
query complexity we recommend Buhrman and de Wolf \cite{BW02}.
\begin{theorem}[Nisan \cite{Nis94}, Buhrman-Cleve-Wigderson \cite{BCW98}]\label{thm:composed_upper}
For any two Boolean functions $f: \{-1,+1\}^n \rightarrow \{-1,+1\}$ and
$g: X \times Y \rightarrow \{-1,+1\}$,
\begin{align*}
D(f \comp g^n)&=O(\DT(f) D(g)) \\
R_{1/4}(f \comp g^n)&=O(\RT_{1/4}(f) R_{1/4}(g) \log \RT_{1/4}(f)) \\
Q_{1/4}(f \comp g^n)&=O(\QT_{1/4}(f) Q_{1/4}(g) \log n).
\end{align*}
\end{theorem}

One advantage of working with block composed functions in light of this upper bound is that
query complexity is in general better understood than communication complexity.  In particular,
a polynomial relationship between deterministic query complexity and degree,
and randomized and quantum query complexities and approximate degree is known.
\begin{theorem}[\cite{NS94,BBCMW01}]
Let $f: \01^n \rightarrow \{-1,+1\}$.  Then
\begin{align*}
\DT(f) =O(\deg(f)^4), \quad \DT(f) =O(\deg_{1/4}(f)^6)
\end{align*}
\end{theorem}

Using this result together with Theorem~\ref{thm:composed_upper} gives the following corollary:
\begin{corollary}\label{cor:deg}
\begin{align*}
D(f \comp g^n) &= O(\deg(f)^4 D(g)), \quad 
R_{1/4}(f \comp g^n) = O(\deg_{1/4}(f)^6 R_{1/4}(g) \log
\deg_{1/4}(f))
\end{align*}
\end{corollary}

Our goal, then, in showing lower bounds on the complexity of a block composed function
$f \comp g^n$ is to get something at least in the ballpark of this upper bound.  Of course, this is not 
always possible --- the protocol given by Theorem~\ref{thm:composed_upper} is not always
optimal.  For example, when $f$ is the PARITY function on $n$ bits, and $g(x,y)=\oplus(x,y)$ this
protocol just gives an upper bound of $n$ bits, when the true complexity is constant. See recent 
results by Zhang \cite{Zha09} and Sherstov \cite{She10} for discussions on the tightness of the 
bounds in Theorem \ref{thm:composed_upper}.

\section{Rank of block composed functions}\label{sec: strongly balanced, rank}
We begin by analyzing the rank of a block composed function $f \comp g^n$ when the inner function
$g$ is strongly balanced.  This case will illustrate the use of the strongly balanced assumption, and
is simpler to understand than the bounded-error situation treated in the next section.

Let us first formally state the definition of strongly balanced.
\begin{definition}[strongly balanced]
Let $A$ be a sign matrix, and $J$ be the all ones matrix of the
same dimensions as $A$.  We say that $A$ is {\em balanced} if
$\Tr(A J^\dag)=0$.  We further say that $A$ is {\em strongly
balanced} if $A J^\dag=A^\dag J=0$.  In words, a sign matrix is strongly
balanced if the sum over each row is zero, and similarly the sum over 
each column is zero. We will say that a two-variable Boolean function is balanced or strongly
balanced if its sign matrix representation is.
\end{definition}

\begin{theorem}
\label{th:rank_compose} Let $f : \{-1,+1\}^n \rightarrow
\{-1,+1\}$ be an arbitrary function, and let $g$ be a strongly
balanced function. Then
$$
\rk(M_{f \comp g^n})=\sum_{\substack{T \subseteq [n],\ \hat f_T \ne 0}} \rk(M_g)^{|T|}.
$$
\end{theorem}

\begin{proof}
Let us write out the sign matrix for $\chi_T \comp g^n$ explicitly.  If we let $M_g^0=J$ be the all 
ones matrix and $M_g^1=M_g$, then we can nicely write the sign matrix representing
$\chi_T(g(x^1, y^1), \ldots, g(x^n,y^n))$ as
$$
M_{\chi_T\circ g^n}=\bigotimes_i M_g^{T[i]}
$$
where $T[i]=1$ if $i \in T$ and $0$ otherwise.

We see that the condition on $g$ implies $M_{\chi_T \comp g^n}
M_{\chi_S \comp g^n}^\dag=0$ if $S \ne T$. Indeed,
\begin{align*}
M_{\chi_T \comp g^n} M_{\chi_S \comp g^n}^\dag & =
\left(\bigotimes_i M_g^{T[i]}\right) \left( \bigotimes_i M_g^{S[i]}\right)^\dag \\
 & =\bigotimes_i \left(M_g^{T[i]} (M_g^{S[i]})^\dag \right) =0.
\end{align*}
This follows since, by the assumption $S \ne T$, there is some $i$ for which $S[i] \ne T[i]$ which 
means that this term is either $M_g J^\dag=0$ or $J M_g^\dag=0$ because $g$ is strongly balanced.
The other case follows similarly.

Now that we have established this property, we can use Fact~\ref{fact:ortho} to obtain
\begin{align*}
 \rk(M_{f \comp g^n}) & =  \rk\Big(\sum_{T \subseteq [n]} \hat f_T \chi_T(g(x^1, y^1), \ldots, g(x^n,y^n))\Big)  \\
 & = \sum_{\substack{T \subseteq [n] \\ \hat f_T \ne 0}} \rk(M_{\chi_T \comp g^n}) \\
 & =\sum_{\substack{T \subseteq [n] \\ \hat f_T \ne 0}} \rk(M_g)^{|T|}
\end{align*}
In the last step we used the fact that rank is multiplicative
under tensor product. 
\end{proof}

Theorem~\ref{th:rank_compose} has the following implication for the log rank conjecture of the composed function with the assumption of the same conjecture for the inner function.
\begin{corollary}
Let $X,Y$ be finite sets, $g: X \times Y \rightarrow \{-1,+1\}$
be a strongly balanced function, and $M_g[x,y]=g(x,y)$ be the
corresponding sign matrix.  Let $f : \{-1,+1\}^n \rightarrow
\{-1,+1\}$ be an arbitrary function. Assume that $\rk(M_g) \ge 2$
and further suppose that there is a constant $c$ such that $D(g)
\le (\log \rk(M_g))^c$. Then
$$
D(f \comp g^n)=O( \log \rk(f \comp g)^{4+c}).
$$
\end{corollary}

\begin{proof}
By Corollary~\ref{cor:deg}, 
$D(f \comp g^n)=O(\deg(f)^4 D(g))=O(\deg(f)^4 (\log \rk(M_g))^c)$. Now, it follows from
Theorem~\ref{th:rank_compose} that $\log \rk(f \comp g) \ge
\deg(f)\log \rk(M_g)$ as by definition of degree there is some $T
\subseteq \01^n$ with $|T|=\deg(f)$ and $\hat f_T \ne 0$.
\end{proof}

In particular, this Corollary means that whenever $g$ is a strongly balanced function on a constant
number of bits and $\rk(M_g) >1$, then the log rank conjecture holds for $f \comp g^n$.  If $g$ is strongly
balanced and $\rk(M_g)=1$ then, up to permutation of rows and columns, which does not change 
the communication complexity, $M_g$ is a tensor product of the XOR function with an all ones matrix.
The log rank conjecture in the case
$(f \comp \oplus^n) (x,y)=f(x_1 \oplus y_1, \ldots, x_n \oplus y_n)$ remains an interesting open question.
Shi and Zhang \cite{SZ09} have recently resolved this question when $f$ is symmetric.

\section{A bound in terms of approximate degree}
\label{sec: strongly balanced, disc}

In this section, we will address the frameworks of Sherstov and Shi-Zhu.  We 
extend both of these frameworks to give more general conditions on the inner function $g$ 
which still imply that the approximate degree of $f$ is a lower bound on the quantum query 
complexity of the composed function $f \comp g^n$.  In outline, both of these frameworks 
follow the same plan.  By Theorem~\ref{thm: approx trace} it suffices to lower bound the 
approximate $\gamma_2$ norm (or even approximate trace norm) of $f \comp g^n$.  To do this, 
they use the  dual formulation given by Proposition~\ref{prop:dual_approximate_norm} 
and construct a witness matrix $B$ which has non-negligible correlation with the target function and small $\gamma_2^*$ (or spectral) norm.

A very nice way to construct this witness, used by both Sherstov and Shi-Zhu, is to use the 
{\em dual polynomial} of $f$.  This is a polynomial $v$ which certifies that the approximate 
polynomial degree of $f$ is at least a certain value.  
More precisely, duality theory of linear programming 
gives the following lemma.
\begin{lemma}[Sherstov 
\cite{She09b}, Shi-Zhu 
\cite{SZ09b}]
Let $f:\{-1,+1\}^n$ $ \rightarrow \{-1,+1\}$ and let $d=\deg_{\epsilon}(f)$.
Then there exists a function $v: \{-1,+1\}^n \rightarrow \R$ such that 
\begin{enumerate}
  \item $\braket{v}{\chi_T}=0$ for every character $\chi_T$ with $|T|< d$.
  \item $\|v\|_1=1$.
  \item $\braket{v}{f} \ge \epsilon$.
\end{enumerate}
\label{dual_poly}
\end{lemma}
Items (2),(3) are used to lower bound the correlation of the witness matrix with the target 
matrix and to upper bound the $\ell_1$ norm of the witness matrix.   In the most difficult step, 
and where these works diverge, Item~(1) is used to upper bound the $\gamma_2^*$ 
(or spectral) norm of the witness matrix.  

We treat each of these frameworks separately in the next two sections.

\subsection{Sherstov's framework}

The proof of the next theorem follows the same steps as Sherstov's proof for pattern matrices 
(Theorem 5.1 \cite{She09b}). Our main contribution is to identify the strongly balanced condition 
as the key property of pattern matrices which enables the proof to work. 

\begin{theorem}
Let $X,Y$ be finite sets, $g: X \times Y \rightarrow \{-1,+1\}$
be a strongly balanced function, and $M_g[x,y]=g(x,y)$ be the
corresponding sign matrix.  Let $f : \{-1,+1\}^n \rightarrow
\{-1,+1\}$ be an arbitrary function.  Then
\begin{equation*}
    Q_{\epsilon}^*(f \comp g^n) \ge \deg_{\epsilon_0}(f) \log_2 \Big(\frac{\sqrt{|X| |Y|}}{\sn{M_g}}
    \Big) - O(1).
\end{equation*}
for any $\epsilon >0$ and $\epsilon_0 > 2\epsilon$.  
\label{thm: block lb}
\end{theorem}
In particular, this result means that the quantum and randomized
complexities of $f \comp g^n$ are polynomially related whenever
$g$ is strongly balanced and $\log
\tfrac{\sqrt{\size(M_g)}}{\sn{M_g}}$ is polynomially related to
the randomized communication complexity of $g$.  While the complexity measure of $g$ used 
here may look strange at first, Shaltiel \cite{Sha03} has shown that it is closely related
to the discrepancy of $g$ under the uniform distribution, as noted above in Theorem~\ref
{thm:shaltiel}.  This theorem strictly generalizes the case of pattern matrices, 
but it could still be the case that the results of Shi-Zhu can
show bounds not possible with this theorem.

\begin{proof}[Proof of Theorem~\ref{thm: block lb}]
Let $d=\deg_{\epsilon_0}(f)$ and let $v$ be a dual polynomial for $f$ with properties as in
Lemma~\ref{dual_poly}.  We define a witness matrix as
$$
B[x,y]=\frac{2^n}{\size(M_g)^n} v(g(x^1, y^1), \ldots, g(x^n, y^n))
$$

Let us first lower bound the inner product $\braket{M_{f \comp g^n}}{B}$.  Notice that as $M_g$ is
strongly balanced, it is in particular balanced, and so the number of ones (or minus ones) in 
$M_g$ is $\size(M_g)/2$.
\begin{align*}
 \braket{M_{f \comp g^n}}{B} 
= \frac{2^n}{\size(M_g)^n}\sum_{z \in \{-1,+1\}^n} f(z) v(z)
\prod_{i=1}^n \Big( \sum_{\scriptsize x^i, y^i: \atop \scriptsize g(x^i,y^i)=z_i} 1\Big) 
= \braket{f}{v} \ge \epsilon_0
\end{align*}
A similar argument shows that $\|B\|_1=1$ as $\|v\|_1=1$.

Now we turn to evaluate $\sn{B}$.  As shown above, the strongly balanced property of $g$ 
implies that the matrices  $\chi_T \comp g^n$ and $\chi_S \comp g^n$ are 
orthogonal for distinct sets 
$S,T \subseteq \01^n$.  We can thus use Fact~\ref{fact:ortho} to compute as follows.
\begin{align*}
\sn{B}&=\frac{2^n}{\size(M_g)^n}\sn{\sum_{T \subseteq [n]} \hat v_T M_{\chi_T \comp g^n}} \\ & =\frac{2^n}{\size(M_g)^n} \max_{T} |\hat v_T| \sn{M_{\chi_T \comp g^n}} \\
&=\max_{T} \ 2^n |\hat v_T| \prod_{i} \frac{\sn{M_g^{T[i]}}}{\size(M_g)} \\
&\le \max_{T: \hat v_T \ne 0} \prod_{i} \frac{\sn{M_g^{T[i]}}}{\size(M_g)} \\
& = \left(\frac{\sn{M_g}}{\sqrt{\size(M_g)}}\right)^d \left(\frac{1}{\size(M_g)} \right)^{n/2}
\end{align*}
In the second to last step we have used that $|\hat v_T| \le 1/2^n$ as $\|v\|_1=1$, and in the
last step we have used the fact that $\sn{J}=\sqrt{\size(M_g)}$.

Now putting everything together we have
\begin{align*}
\frac{\trn{M_{f \comp g^n}}^{\epsilon_0}}{\sqrt{\size(M_{f \comp g^n})}} \ge
\frac{1}{12} \left(\frac{\sqrt{\size(M_g)}}{\sn{M_g}}\right)^d
\end{align*}
The lower bound on quantum communication complexity now follows from 
Theorem~\ref{thm: approx trace}.
\end{proof}

Using the theorem of Shaltiel relating discrepancy to the spectral norm 
Theorem~\ref{thm:shaltiel}, we get the following corollary:
\begin{corollary}\label{cor:disc}
Let the quantities be defined as in Theorem~\ref{thm: block lb}.
$$
Q_{1/8}^*(f \comp g^n) \geq
\frac{1}{3}\deg_{1/3}(f) \big(\log \big(\frac{1}{\disc_U(M_g)}\big)-7 \big) -O(1).
$$
\end{corollary}

\noindent {\bf Comparison to Sherstov's pattern matrix}: As mentioned in \cite{She09b}, Sherstov's pattern matrix method can prove quantum lower bound of $\Omega(\deg_{\epsilon}(f))$ for block composed functions $f\circ g^n$ if the matrix $M_g$ contains the following $4\times 4$ one as a submatrix:
\begin{equation*}
	S_4 = 
	\begin{bmatrix}
	 1 & -1 & 1 & -1 \\
   1 & -1 & -1 & 1 \\
   -1 & 1 & 1 & -1 \\
   -1 & 1 & -1 & 1
	\end{bmatrix}
\end{equation*}
In this paper we show that the same lower bound holds as long as $M_g$ contains a 
strongly balanced submatrix or rank greater than one. Are there strongly balanced matrices not 
containing $S_4$ as a submatrix? It turns out that the answer is yes: we give the following 
$6\times 6$ matrix as one example. 
\begin{equation*}
	S_6 = 
	\begin{bmatrix}
		1 & 1 & 1 & -1 & -1 & -1  \\
		1 & 1 & -1 & 1 & -1 & -1  \\
		1 & -1 & -1 & -1 & 1 & 1  \\ 
		-1 & -1 & 1 & 1 & 1 & -1  \\ 
		-1 & 1 & -1 & -1 & 1 & 1  \\ 
		-1 & -1 & 1 & 1 & -1 & 1
	\end{bmatrix}
\end{equation*}

\subsection{Shi-Zhu framework}
The method of Shi-Zhu does not restrict the form of the inner function $g$, but 
rather works for any $g$ which is sufficiently ``hard.''  The hardness condition they 
require is phrased in terms of a somewhat awkward measure they term spectral discrepancy.  

\begin{definition}[spectral discrepancy]
Let $A$ be a $m$-by-$n$ sign matrix.  
The spectral discrepancy of $A$, denoted $\rho(A)$, is the smallest $r$ such that there is 
a submatrix $A'$ of $A$ and a probability distribution $\mu$ on the entries of $A'$ satisfying:
\begin{enumerate}
  \item $A'$ is balanced with respect to $\mu$, i.e.\ the distribution which gives equal weight to 
  $-1$ entries and $+1$ entries of $A'$.
   \item The spectral norm of $A' \hp \mu$ is small: 
   $$\sn{A' \hp \mu} \le \frac{r}{\sqrt{\size(A')}}$$
  \item The entrywise absolute value of the matrix $A' \hp \mu$ should also have a bound on its 
  spectral norm in terms of $r$: 
  $$\sn{|A' \hp \mu|} \le \frac{1+r}{\sqrt{\size(A')}}$$
\end{enumerate}
\end{definition}
While conditions (1),(2) in the definition of spectral discrepancy are quite natural, 
condition~(3) can be complicated to verify.  Note that condition~(3) will always 
be satisfied when $\mu$ is taken to be the uniform distribution.
Using this notion of spectral discrepancy, Shi-Zhu show the following theorem.
\begin{theorem}[Shi-Zhu \cite{SZ09b}]
Let $f : \{-1,+1\}^n \rightarrow \{-1,+1\}$, and $g: X \times Y
\rightarrow \{-1,+1\}$.  For any $\epsilon$ and $\epsilon_0 > 2\epsilon$,
$$
Q_{\epsilon}(f \comp g^n) \ge \Omega(\deg_{\epsilon_0}(f)).
$$
provided $\rho(M_g) \le \tfrac{\deg_{\epsilon_0}(f)}{2en}$.
Here $e=2.718\ldots$ is Euler's number.
\end{theorem}

Chattopadhyay \cite{Cha08} extended the technique of Shi-Zhu to the case of multiparty 
communication complexity, answering an open question of Sherstov \cite{She08c}.  In doing 
so, he gave a more natural condition on the hardness of $g$ in terms of an upper bound on 
discrepancy frequently used in the multiparty setting and originally 
due to Babai, Nisan, and Szegedy \cite{BNS92}.  
As all that is crucially needed is subadditivity, we do the argument here with $\gamma_2^*$, 
which is essentially equal to the discrepancy.

\begin{theorem}
Let $f : \{-1,+1\}^n \rightarrow \{-1,+1\}$, and $g: X \times Y
\rightarrow \{-1,+1\}$.
Fix $0 < \epsilon < 1/2$, and let $\epsilon_0 > 2\epsilon$.
Then
$$
Q_\epsilon^*(f \comp g^n) \ge \deg_{\epsilon_0}(f) -O(1).
$$
provided there is a distribution $\mu$ which is balanced with respect to $g$ and
for which $\gamma_2^*(M_g \hp \mu) \le \tfrac{\deg_{\epsilon_0}(f)}{2en}$.
\label{thm:shi-zhu}
\end{theorem}

\begin{proof}
We again use Proposition~\ref{prop:dual_approximate_norm}, this time with the 
$\gamma_2$ norm instead of the trace norm.
$$
\gamma_2^{\epsilon_0}(M_{f \comp g^n}) = \max_B
\frac{\braket{M_{f \comp g^n}}{B} -\epsilon_0 \|B\|_1}{\gamma_2^*(B)}.
$$
To prove a lower bound we choose a witness matrix $B$ as follows
$$
B[x,y]=2^n \cdot v(g(x^1,y^1), \ldots, g(x^n,y^n)) \cdot \prod_{i=1}^n \mu(x^i,y^i).
$$
where $v$ witnesses that $f$ has approximate degree at least $d=\deg_{\epsilon_0}(f)$.  
This definition is the same as in the previous section where $\mu$ was 
simply the uniform distribution.  As argued before, we have 
$\braket{M_{f \comp g^n}}{B} \ge \epsilon_0$ and $\|B\|_1=1$ because 
$M_g \hp \mu$ is balanced.  

We again expand $B$ as
$$
B=2^n \sum_{T: |T| \ge d} \hat v_T \bigotimes_{i=1}^n (M_g \hp \mu)^{T(i)},
$$
where $(M_g \hp \mu)^1=M_g \hp \mu$ and $(M_g \hp \mu)^0=\mu$.

Now comes the difference with the previous proof.  As we do not have special knowledge 
of the function $g$, we simply bound $\gamma_2^*(B)$ using the triangle inequality.
\begin{align*}
\gamma_2^*(B) &\le 2^n
\sum_{T: |T| \ge d} |\hat v_T| \ \gamma_2^*\left(\bigotimes_{i=1}^n (M_g \hp \mu)^{T(i)} \right)\\
&= 2^n \sum_{T: |T| \ge d} |\hat v_T| \ \gamma_2^*(M_g \hp \mu)^{|T|} \gamma_2^*(\mu)^{n-|T|} \\
&\le \sum_{T: |T| \ge d} \gamma_2^*(M_g \hp \mu)^{|T|},
\end{align*}
where in the last step we have used that $\gamma_2^*(\mu) \le 1$ as $\mu$ is a probability
distribution and that $|\hat v_T| \le 2^{-n}$. In the second step (equality) we used the fact
that $\gamma_2^*$ is multiplicative with respect to tensor product, a property proved
in \cite{LSS08}.
We continue with simple arithmetic:
\begin{align*}
\gamma_2^*(B) &\le \sum_{i=d}^n {n \choose i} \gamma_2^*(M_g \hp \mu)^i \\
&\le \sum_{i=d}^n \left( \frac{en \gamma_2^*(M_g \hp \mu)}{d} \right)^i \\
&\le 2^{-d}
\end{align*}
provided that $\gamma_2^*(M_g \hp \mu) \le \frac{d}{2en}$.
\end{proof}

\section{A general framework for functions composed through a group} \label{sec:group}
In this section we begin the study of more general function composition
through a group $G$.  In this case the outer function $f: G \rightarrow \{-1,+1\}$
is a class function, i.e.\ invariant on conjugacy classes, and the inner function
$g: X \times Y \to G$ has range $G$.  We define the composed function as
$F(x,y)=f(g(x,y))$.  In previous sections of the paper we have just dealt with the case
$G=\Z_2^n$. 

Let us recall the basic idea of the proof of Theorem~\ref{thm: block lb}.  To prove a lower bound on 
the quantum communication complexity for a composed function $f \comp g$, we constructed a 
witness 
matrix $B$ which had non-negligible correlation with $f \comp g$ and small spectral norm.  
To do this, following the work of Sherstov and
Shi-Zhu \cite{She09b,SZ09b}, we considered the dual polynomial $p$ of $f$ using LP duality.  The 
dual polynomial has two important properties, first that $p$ has non-negligible correlation with $f$ 
and 
second that $p$ has no support on low degree polynomials.  We can then use the first property to 
show that the composed function $p \comp g$ will give non-negligible inner 
product with $f \comp g$ and the second to upper bound the spectral norm of $p \comp g$.  
The second of these tasks is the more difficult. In the case of $G = \{-1,+1\}^n$, the degree of a 
character $\chi_T$ is a natural
measure of how ``hard'' the character is --- the larger $T$ is, the smaller the spectral norm of 
$\chi_T \comp g$ will be.  In the general group case, however, it is less clear what the corresponding ``hard'' and ``easy'' characters should be. In Section \ref{sec: general functions}, we will show that this framework actually works for an arbitrary partition of the basis functions into Easy and Hard.  That is for any arbitrary partition of the basis functions into Easy and Hard sets we can follow the plan outlined above and look for a function with support on the Hard set which has non-negligible correlation with $f$.  

In carrying out this plan, one is still left with upper bounding $\|M_{p\circ g}\|$.  Here, as in the Boolean case, it as again very convenient to have an orthogonality condition which can greatly 
simplify the computation of $\sn{M_{p \comp g}}$ and give good bounds.
In the Boolean case we have shown that $M_g$ being strongly balanced implies this key orthogonality
condition. In Section \ref{sec: group} and \ref{sec: abelian group}, we will show that for the general group, the condition is not only about each row and column of matrix $M_g$, but all pairs of rows and pairs of columns.  In the Abelian group case, this reduces to a nice group invariance condition. 

Even after applying the orthogonality condition to use the maximum bound instead of the triangle inequality for $\|M_{p\circ g}\|$, the remaining term $\|M_{\chi_i \comp g}\|$ (where $\chi_i$ is a ``hard" character) is still not easy to upper bound. For block composed functions, fortunately, the tensor structure makes it feasible to compute. Section \ref{sec: block} gives a generalized version of Theorem \ref{thm: block lb}.

\subsection{General framework}\label{sec: general functions}
For a multiset $T$, $x\in T$ means $x$ running over $T$. Thus $T = \{a(s): s\in S\}$ means the multiset formed by collecting $a(s)$ with $s$ running over $S$. 

For a set $S$, denote by $L_{\mathbb C}(S)$ the $|S|$-dimensional vector space over the field $\mbC$ (of complex numbers) consisting of all linear functions from $S$ to $\mbC$, endowed with inner product $\langle \psi, \phi \rangle = \frac{1}{|S|}\sum_{s\in S} \psi(s)\overline{\phi(s)}$. The distance of a function $f\in L_{\mathbb C}(S)$ to a subspace $\Phi$ of $L_{\mathbb C}(S)$, denoted by $d(f,\Phi)$, is defined as $\min\{\delta: \|f'-f\|_\infty \leq \delta, f'\in \Phi\}$, \ie the magnitude of the least entrywise perturbation to turn $f$ into $\Phi$.



In the above setting, Theorem~\ref{thm: block lb} generalizes to the following. 

\begin{theorem}\label{thm: lb, general function}
    Consider a sign matrix $A = [f(g(x,y))]_{x,y}$ where $g: X\times Y \rightarrow S$ for a set $S$, and $f: S\rightarrow \{-1, +1\}$. Suppose that we can find an orthogonal basis functions $\Psi = \{\psi_i: i\in [|S|]\}$ for $L_{\mathbb C}(S)$. For any hardness partition $\Psi = \Psi_{Hard} \uplus \Psi_{Easy}$, let $\delta = d(f, span(\Psi_{Easy}))$. If
    \begin{enumerate}
        \item {\text{\bf(regularity)}} The multiset $\{g(x,y): x\in X, y\in Y\}$ is a multiple of $S$, \ie $S$ repeated for some number of times.
        \item {\text{\bf(orthogonality)}} for all $x,x',y,y'$ and all distinct $\psi_i,\psi_j \in \Psi_{Hard}$,
        \begin{align*}
            \sum_y \psi_i(g(x,y))\overline{\psi_j(g(x',y))} =  \sum_x \psi_i(g(x,y))\overline{\psi_j(g(x,y'))} = 0,
        \end{align*}
    \end{enumerate}
    then
    \begin{align*}
        Q_\epsilon(A) \geq \log_2 \frac{\sqrt{MN}\cdot (\delta -2\epsilon)}{\max_{\psi_i\in \Psi_{Hard}} (\max_g |\psi_i(g)| \cdot \|[\psi_{i}(g(x,y))]_{x,y}\|) } - O(1).
    \end{align*}
\end{theorem}
Using the idea of finding a certificate of the high approximate degree by duality \cite{She09b,SZ09b}, we have the following fact analogous to Lemma \ref{dual_poly}.
\begin{lemma}\label{lem: duality}
For a function $f: S\rightarrow \mbC$ and a subspace $\Phi$ of $L_{\mathbb C}(S)$, if $d(f,span(\Phi)) = \delta$, then there exists a function $h$ \st
\begin{align}
    &\label{eq: dual fn, hard} \hat h_i = 0, \ \forall \psi_i\in \Phi
    \\
    &\label{eq: dual fn, 1-norm} \sum_{g\in G} |h(g)| \leq 2,
    \\
    &\label{eq: dual fn, correlation} |\sum_{g\in G} f(g) \overline{h(g)}| > \delta
\end{align}
\end{lemma}
Using the lemma, we can prove the Theorem \ref{thm: lb, general function}. 
\begin{proof} (of Theorem \ref{thm: lb, general function})
    By the regularity property, we know that when $(x,y)$ runs over $X\times Y$, $g(x,y)$ runs over $S$ exactly $K$ times where $K = MN/|G|$. Consider $B = \frac{1}{K}[h(g(x,y))]_{x,y}$, $h$ obtained by Lemma \ref{lem: duality}; we want to apply Proposition \ref{prop:dual_approximate_norm} and Theorem \ref{thm: approx trace} by using this $B$. First,
\begin{equation}\label{eq: 1-norm of B}
    \|B\|_1 = \frac{1}{K}\sum_{x,y}|h(g(x,y))| = \sum_{g\in G} |h(g)| \leq 1.
\end{equation}
Also,
\begin{align*}
    |\langle A, B \rangle| & = \frac{1}{K} |\sum_{x,y} f(g(x,y))\overline{h(g(x,y))}| \\
    & = |\sum_{g\in G} f(g)\overline{h(g)}| > \delta
\end{align*}
Now we need to compute \[\|B\| = \frac{1}{K}\big\|[\sum_{\chi_i\in Hard} \hat h_i \chi_i(g(x,y))]_{x,y}\big\|.\] Note that
\begin{align*}
    & [\psi_i(g(x,y))]_{x,y}^\dag [\psi_j(g(x,y))]_{x,y} 
    =  [\sum_x \overline{\psi_i(g(x,y))}\psi_j(g(x,y'))]_{y, y'}
\end{align*}
and
\begin{align*}
    & [\psi_i(g(x,y))]_{x,y} [\psi_j(g(x,y))]_{x,y}^\dag 
    =  [\sum_y \psi_i(g(x,y))\overline{\psi_j(g(x',y))}]_{x, x'}.
\end{align*}
Thus the orthogonality condition implies that
\begin{align*}
    & [\psi_i(g(x,y))]_{x,y}^\dag [\psi_j(g(x,y))]_{x,y} 
    =  [\psi_i(g(x,y))]_{x,y} [\psi_j(g(x,y))]_{x,y} ^\dag = 0
\end{align*}
for all $i\neq j$. Now as in \cite{She09b}, we can use the max bound
\begin{align*}
    \|B\| & = \frac{1}{K}\max_{i: \psi_i\in \Psi_{Hard}} \big\|\hat h_{i} [\psi_{i}(x,y)]_{x,y}\big\| \\
    & \leq \frac{1}{K}\max_{i: \psi_i\in Hard} |\hat h_{i}| \max_{\psi_i\in Hard} \big\| [\psi_{i}(x,y)]_{x,y}\big\| \\
    & \leq \frac{1}{K|G|} \max_{\psi_i\in Hard} \big( \max_{g} |\psi_i(g)| \cdot \big\|[\psi_{i}(x,y)]_{x,y} \big\|\big).
\end{align*}
where the last inequality is due to Eq. \eqref{eq:l-infty bound of F(f)} and Eq. \eqref{eq: 1-norm of B}. Finally note $K|G| = MN$ to complete the proof.
\end{proof}

In the Boolean block composed function case, the regularity condition reduces to the matrix $[g(x,y)]$ being balanced, and later we will prove that the orthogonality condition reduces to the strongly balanced property. From this theorem we can see that the way to partition $\Psi$ into $\Psi_{Easy}$ and $\Psi_{Hard}$ does not really matter for the lower bound proof passing through. However, the partition does play a role when we later bound the spectral norm in the denominator.


\subsection{Functions with group symmetry}\label{sec: group}
For a general finite group $G$, two elements $s$ and $t$ are conjugate, denoted by $s\sim t$, if there exists an element $r\in G$ \st $rsr^{-1} = t$. Define $H$ as the set of all class functions, \ie functions $f$ \st $f(s) = f(t)$ if $s\sim t$. Then $H$ is an $h$-dimensional subspace of $L_{\mathbb C} (G)$, where $h$ is the number of conjugacy classes. The irreducible characters $\{\chi_i: i\in [h]\}$ form an orthogonal basis of $H$. For a class function $f$ and irreducible characters $\chi_i$, denote by $\hat f_i$ the coefficient of $\chi_i$ in expansion of $f$ according to $\chi_i$'s, \ie $\hat f_i = \langle \chi_i, f \rangle =  \frac{1}{|G|}\sum_{g\in G} \chi_i(g)\overline{f(g)}$. An easy fact is that for any $i$, we have
\begin{align}\label{eq:l-infty bound of F(f)}
	 |\hat f_i| & = \frac{1}{|G|}\left|\sum_{g\in G} \chi_i(g) \overline{f(g)}\right| \leq \frac{1}{|G|}\sum_{g\in G} |f(g)| |\chi_i(g)| \leq \Big(\frac{1}{|G|} \sum_{g\in G}|f(g)|\Big) \cdot \max_{g} |\chi_i(g)|.
\end{align}
If $G$ is Abelian, then it always has $|\chi_i(g)| = 1$, thus $\max_i |\hat f_i| \leq \frac{1}{|G|} \sum_{g\in G}|f(g)|$. For general groups, we have $|\chi_i(g)| \leq \deg(\chi_i)$, where $\deg(\chi_i)$ is the degree of $\chi_i$, namely the dimension of the associated vector space.

In this section we consider the setting that $S$ is a finite group $G$. The goal is to exploit properties of group characters to give better form of the lower bound. In particular, we hope to see when the second condition holds and what the matrix operator norm $\|[\psi_(g(x,y))]_{x,y}\|$ is  in this setting.

The standard orthogonality of irreducible characters says that $\sum_{s\in G} \chi_i(s) \overline{\chi_j(s)} = 0$. The second condition in Theorem \ref{thm: lb, general function} is concerned with a more general case: For a multiset $T$ with elements in $G\times G$, we need
\begin{equation}\label{eq: orthogonal graph requirement}
    \sum_{(s, t)\in T} \chi_i(s) \overline{\chi_j(t)} = 0, \qquad \forall i\neq j.
\end{equation}
The standard orthogonality relation corresponds to the special that $T = \{(s,s): s\in G\}$. We hope to have a characterization of a multiset $T$ to make Eq. \eqref{eq: orthogonal graph requirement} hold.

We may think of the a multiset $T$ with elements in set $S$ as a function on $S$, with the value on $s\in S$ being the multiplicity of $s$ in $T$. Since characters are class functions, for each pair $(C_k, C_l)$ of conjugacy classes, only the value $\sum_{g_1\in C_k, t\in C_l} T(g_1, t)$ matters for the sake of Eq. \eqref{eq: orthogonal graph requirement}. We thus make $T$ a class function by taking average within each class pair $(C_k, C_l)$. That is, define a new function $T'$ as
\begin{equation*}
    T'(s, t) = \sum_{s\in C_k, t\in C_l} T(s, t) / (|C_k||C_l|), \ \forall s\in C_k,\ \forall t\in C_l.
\end{equation*}

\begin{proposition}\label{prop: general group condition}
    For a finite group $G$ and a multiset $T$ with elements in $G\times G$, the following three statements are equivalent:
    \begin{enumerate}
        \item $\sum_{(s, t)\in T} \chi_i(s)\overline{\chi_j(t)} = 0, \ \forall i\neq j$
        \item $T'$, as a function, is in $span\{\chi_i \otimes \overline{\chi_i}: i\in[h]\}$
        \item $[T'(s,t)]_{s,t} = C^\dagger D C$ where $D$ is a diagonal matrix and $C = [\chi_i(s)]_{i,s}$. That is, $T'$, as a matrix, is normal and diagonalized exactly by the irreducible characters.
    \end{enumerate}
\end{proposition}
\begin{proof}
    Let $H_2$ be the subspace consisting functions $f: G\times G \rightarrow \mbC$ \st $f(s, t) = f(s', t')$ if $s \sim s'$, $t \sim t'$. Note that for direct product group $G\times G$, $\{\chi_i\otimes \overline{\chi_j}: i,j\}$ form an orthogonal basis of $H_2$:
    \begin{align*}
        & \sum_{s, t\in G} \chi_{i}(s)\overline{\chi_{j}(t)}\overline{\chi_{i'}(s)\overline{\chi_{j'}(t)}} =  \big(\sum_{s\in G}\chi_{i}(s)\overline{\chi_{i'}(s)}\big) \big(\sum_{t\in G}\overline{\chi_{j}(t)}\chi_{j'}(t)\big) = 0
    \end{align*}
    unless $i = i'$ and $j = j'$. Note that by viewing $T$ as a function from $G\times G$ to $\mbC$, the Eq. \eqref{eq: orthogonal graph requirement} and the definition of $T'$ imply that
    \begin{equation*}
        \langle \chi_i\otimes \overline{\chi_j},\ T \rangle = 0, \forall i\neq j
    \end{equation*}
    Thus the first two statements are equivalent.

    Note that
    \begin{align*}
        T'\in span\{\chi_i \otimes \overline{\chi_i}: i\in[h]\}
        \quad \Leftrightarrow  \quad T'(s, t) = \sum_i \alpha_i \chi_i(s) \overline{\chi_i(t)} \quad \text{ for some  $\alpha_i$'s}
    \end{align*}
    Denote by $C_{h\times |G|} = [\chi_i(g)]_{i,g}$ the matrix of the character table. Then observe that the summation in the last equality is nothing but the $(s, t)$ entry of the matrix $C^\dag diag(\alpha_1, \cdots, \alpha_h) C$. Therefore the equivalence of the second and third statements follows.
\end{proof}

\subsection{Abelian group}\label{sec: abelian group}
When $G$ is Abelian, we have further properties to use. The first one is that $|\chi_i(g)| = 1$ for all $i$. The second one is that the irreducible characters are homomorphisms of $G$; that is, $\chi_i(st) = \chi_i(s)\chi_i(t)$. This gives a clean characterization of the orthogonality condition by group invariance. For a multiset $T$, denote by $sT$ another multiset obtained by collecting all $st$ where $t$ runs over $T$. A multiset $T$ with elements in $G\times G$ is \emph{$G$ invariant} if it  satisfies $(g,g)T = T$ for all $g\in G$. We can also call a function $T: G\times G\rightarrow \mbC$ $G$ invariant if $T(s,t) = T(rs,rt)$ for all $r,s,t\in G$. The overloading of the name is consistent when we view a multiset $T$ as a function (counting the multiplicity of elements).

\begin{proposition}\label{prop: Abelian group invariant condition}
    For a finite Abelian group $G$ and a multiset $T$ with elements in $G\times G$,
    \begin{equation}\label{eq: orthogonal graph, Abelian}
        \text{$T$ is $G$ invariant} \Leftrightarrow \sum_{(s, t)\in T} \chi_i(s)\overline{\chi_j(t)} = 0, \quad \forall i\neq j.
    \end{equation}
\end{proposition}
\begin{proof}
    $\Rightarrow$: Since $T$ is $G$ invariant, $T = (r,r)T$ and thus,
    \begin{align}
        \sum_{(s, t)\in T} \chi_i(s)\overline{\chi_j(t)} & = \sum_{(s, t)\in (r,r)T} \chi_i(s)\overline{\chi_j(t)} \\
        & = \sum_{(s', t')\in T} \chi_i(rs')\overline{\chi_j(rt')}
    \end{align}
    Now using the fact that irreducible characters of Abelian groups are homomorphisms, we have
    \begin{align*}
        \sum_{(s', t')\in T} \chi_i(rs')\overline{\chi_j(rt')}
        = & \sum_{(s', t')\in T} \chi_i(r) \chi_i(s')\overline{\chi_j(r)}\overline{\chi_j(t')} \\
        = & \chi_i(r)\overline{\chi_j(r)} \Big(\sum_{(s', t')\in T} \chi_i(s')\overline{\chi_j(t')}\Big)
    \end{align*}
    But note that this holds for any $r\in G$, thus also for the average of them. That is,
    \begin{align*}
        \sum_{(s, t)\in T} \chi_i(s)\overline{\chi_j(t)} = & \frac{1}{|G|} \Big(\sum_{r\in G} \chi_i(r)\overline{\chi_j(r)}\Big) \Big(\sum_{(s', t')\in T} \chi_i(s')\overline{\chi_j(t')}\Big) = 0,
    \end{align*}
    by the standard orthogonality property of different irreducible characters.

    $\Leftarrow$: Since $\sum_{(s, t)\in T} \chi_i(s)\overline{\chi_j(t)} = 0$, $\forall i\neq j$, we know that $T$ as a function is in $span\{\chi_i \otimes \overline{\chi_i}: i\}$. Note that any linear combination of $G$ invariant functions is also $G$ invariant. Thus it remains to check that each basis $\chi_i \otimes \overline{\chi_i}$ is $G$ invariant, which is easy to see:
    \begin{equation*}
        \chi_i(rs) \overline{\chi_i(rt)} = \chi_i(r)\chi_i(s)  \overline{\chi_i(r)}\overline{\chi_i(t)} =\chi_i(s) \overline{\chi_i(t)}.
    \end{equation*}
    This finishes the proof.
\end{proof}

Another nice property of Abelian groups is that the orthogonality condition condition implies the regularity one.
\begin{proposition}\label{prop: Abelian ortho reg}
    For an Abelian group $G$, if either $T^{y,y}$ is $G$ invariant for all $y$ or $S^{x,x}$ is $G$ invariant for all $x$, then $G|\{g(x,y): x\in X, y\in Y\}$.
\end{proposition}
\begin{proof}
    Note that $T^{y,y}(s,s) = |\{x:g(x,y) = s\}|$, thus $T^{y,y}$ being $G$ invariant implies that $|\{x: g(x,y) = s\}| = |\{x: g(x,y) = t\}|$ for all $s,t\in G$. Thus the column $y$ in matrix $[g(x,y)]_{x,y}$, when viewed as a multiset, is equal to $G$ repeated $|Y|/|G|$ times. Therefore the whole multiset $\{g(x,y): x\in X, y\in Y\}$ is a multiple of $G$ as well.
\end{proof}

What we finally get for Abelian groups is the following.
\begin{corollary}
    For a sign matrix $A = [f(g(x,y))]_{x,y}$ and an Abelian group $G$, if $d(f,span(Ch_{Easy})) = \Omega(1)$, and the multisets $S^{x,x'} = \{(g(x,y), g(x',y)): y\in Y\}$ and $T^{y,y'} = \{(g(x,y), g(x,y')): x\in X\}$ are $G$ invariant for any $(x,x')$ and any $(y,y')$, then
    \begin{equation*}
        Q(A) \geq \log_2 \frac{\sqrt{MN}}{\max_{i\in Hard} \|[\chi_i(g(x,y))]_{x,y}\|} - O(1).
    \end{equation*}
\end{corollary}

\subsection{Block composed functions}\label{sec: block}
We now consider a special class of functions $g$: block composed functions. Suppose the group $G$ is a product group $G = G_1 \times \cdots \times G_t$, and $g(x,y) = (g_1(x^1, y^1), \cdots, g_t(x^t, y^t))$ where $x = (x^1, \cdots, x^t)$ and $y = (y^1, \cdots, y^t)$. That is, both $x$ and $y$ are decomposed into $t$ components and the $i$-th coordinate of $g(x,y)$ only depends on the $i$-th components of $x$ and $y$. The tensor structure makes all the computation easy. Theorem \ref{thm: block lb} can be generalized to the general product group case for arbitrary groups $G_i$.
\begin{definition}
	The $\epsilon$-approximate degree of a class function $f$ on product group $G_1\times \cdots \times G_t$, denoted by $d_\epsilon(f)$, is the minimum $d$ \st $\|f-f'\|_\infty \leq \epsilon$, where $f'$ can be represented as a linear combination of irreducible characters with at most $d$ non-identity component characters. 
\end{definition}

\begin{theorem}
    For sign matrix \[A = [f(g_1(x^1,y^1), \cdots, g_t(x^t,y^t)]_{x,y}\] where all $g_i$ satisfy their orthogonality conditions, we have
    \begin{equation*}
        Q(A) \ge \min_{\{\chi_i\}, S} \sum_{i\in S} \log_2 \frac{\sqrt{\size(M_{g_i})}}{\deg(\chi_i)\|M_{\chi_i \comp g_i}\|} - O(1)
    \end{equation*}
    where the minimum is over all $S\subseteq [n]$ with $|S| > \deg_{1/3}(f)$, and all non-identity irreducible characters $\chi_i$ of $G_i$.
\end{theorem}

\begin{proof}
Recall that an irreducible character $\chi$ of $G$ is the tensor product of irreducible characters $\chi_i$ of each component group $G_i$. Let Hard be the set of irreducible characters $\chi$ with more than $d$ non-identity component characters. Fix a hard character $\chi$, and denote by $S$ the set of coordinates of its non-identity characters. 
\begin{align*}
	\|[\chi (g(x,y))]\| & = \|\bigotimes_{i\in [t]} [\chi_i (g_i(x^i,y^i))]\| \\
	& = \prod_{i\in [t]} \| [\chi_i (g_i(x^i,y^i))]\| \\
	& = \prod_{i\in S} \| [\chi_i (g_i(x^i,y^i))]\| \times \prod_{i\notin S} \|J_{|X_i|\times |Y_i|}\| \\
	& = \prod_{i\in S} \| M_{\chi_i \comp g_i}\| \times \prod_{i\notin S} \sqrt{\size(M_{g_i})}
\end{align*}
Thus by Theorem \ref{thm: lb, general function} and Eq. \eqref{eq:l-infty bound of F(f)}, we have
\begin{align}
	Q(A) \geq \log_2 \prod_{i\in S} \frac{\sqrt{\size(M_{g_i})}}{\deg(\chi_i)\cdot \|M_{\chi_i\circ g_i}\|} - O(1)
\end{align}
proving the theorem. 
\end{proof}

Previous sections as well as \cite{SZ09b} consider the case where all $g_i$'s are the same and all $G_i$'s are $\mbZ_2$. In this case, the above bound is equal to the one in Theorem \ref{thm: block lb}, and the following proposition says that the group invariance condition degenerates to the strongly balanced property.
\begin{proposition}\label{prop: degeneration}
    For $G = \mbZ_2^{\times t}$, the following two conditions for $g = (g_1, \cdots, g_t)$ are equivalent:
    \begin{enumerate}
        \item The multisets $S^{x,x'} = \{(g(x,y), g(x',y)): y\in Y\}$ and $T^{y,y'} = \{(g(x,y), g(x,y')): x\in X\}$ are $G$ invariant for any $(x,x')$ and any $(y,y')$,
        \item Each matrix $[g_i(x^i,y^i)]_{x^i, y^i}$ is strongly balanced.
    \end{enumerate}
\end{proposition}
\begin{proof}
    1 $\Rightarrow$ 2: $S^{x,x'}$ being $G$ invariant implies that for all $\{z_i\}, \{u_i\}, \{v_i\}$,
    \begin{align*}
        & |\{y: z_i g_i(x^i, y^i) = u_i,\ z_i g_i(x'^i, y^i) = v_i, \forall i\}| \\
        = & |\{y: g_i(x^i, y^i) = u_i,\ g_i(x'^i, y^i) = v_i, \forall i\}|.
    \end{align*}
    Take $x' = x$ and $u = v$. Now for each $i$ and each row $x^i$, take $z_i = -1$ (where the group $\mbZ_2$ is represented by $\{\pm 1\}$). For all other $i'\neq i$, take $z_{i'} = 1$. This assignment will show that
    \begin{equation*}
        |\{y: g_i(x^i, y^i) = - u_i\}| = |\{y: g_i(x^i, y^i) = u_i\}|.
    \end{equation*}
    That is, the row $x^i$ in matrix $[g_i(x^i,y^i)]$ is balanced. Similarly we can show the balance for each column.

    2 $\Rightarrow$ 1: It is enough to show that for each $i$ and each $\{z_i\}, \{u_i\}, \{v_i\}$,
    \begin{align*}
        & |\{y^i: z_i g_i(x^i, y^i) = u_i,\ z_i g_i(x'^i, y^i) = v_i\}| \\
        = & |\{y^i: g_i(x^i, y^i) = u_i,\ g_i(x'^i, y^i) = v_i\}|.
    \end{align*}
    First consider the case $x'^i = x^i$. If $u_i \neq v_i$ then both numbers are 0; if $u_i = v_i$ then both numbers are $|\{y^i\}|/2$ by the balance of row $x^i$. Now assume $x'^i \neq x^i$. Denote $a_{bb'} = |\{y^i: g_i(x^i, y^i) = b,\ g_i(x^i, y^i) = b'\}|$, then the above requirement amounts to $a_{00} = a_{11}$ and $a_{01} = a_{10}$.

    Note that we have
    \begin{align*}
        a_{00} + a_{01} & = |\{y^i: g_i(x^i, y^i) = 0\}| = |\{y^i: g_i(x^i, y^i) = 1\}| = a_{10} + a_{11}
    \end{align*}
    where the second equality is due to the balance of row $x^i$. And similarly we have $a_{00} + a_{10} = a_{01} + a_{11}$ by balance of row $x'^i$. Combining the two, we get $a_{00} = a_{11}$ and $a_{01} = a_{10}$ as desired.
\end{proof}

It is worth noting that the conclusion does not hold if any group $G_i$ with size larger than two. We omit the counterexamples here.

\bibliographystyle{alpha}
\bibliography{complexity}

\end{document}